\documentclass[letterpaper]{article}
\usepackage{graphicx} 
\usepackage{biblatex}
\usepackage{amsmath}
\usepackage{amsthm}
\newtheorem{proposition}{Proposition}

\title{An additively optimal interpreter for approximating Kolmogorov prefix complexity}
\author{%
Zoe Leyva-Acosta,\\%
Eduardo Acuña Yeomans and \\%
Francisco Hernandez-Quiroz%
}

\begin{document}

\maketitle

\begin{abstract}
    We study practical approximations to Kolmogorov prefix complexity~($K$) using IMP2, a high-level programming language. Our focus is on investigating the interpreter optimality for this language as the reference machine for the Coding Theorem Method~(CTM). A method advanced to deal with applications to algorithmic complexity different to the popular traditional lossless compression approach based on the principles of algorithmic probability. The chosen model of computation is proven to be suitable for this task and a comparison to other models and methods is performed. Our findings show that CTM approximations using our model do not always correlate with results from lower-level models of computation. This suggests some models may require a larger program space to converge to Levin's universal distribution. Furthermore, we compare CTM with an upper bound to Kolmogorov complexity and find a strong correlation, supporting CTM's validity as an approximation method with finer-grade resolution of~$K$.
\end{abstract}

\section{Introduction}

The algorithmic Kolmogorov complexity of a string is the length of a smallest program producing such string when executed by a reference machine~\cite{KolmogorovOriginal}. Apart from being an uncomputable measure, another inconvenience about algorithmic complexity is that in practice the choice of reference model or description method plays a role in approximating calculations. In that sense, the Invariance theorem~\cite{IKCA} tells us about the existence of some \emph{optimal} universal machines such that, when chosen as reference model, produce an invariant measure, up to an additive constant. This constant can skew Kolmogorov complexity's values of short strings, as the invariance is, even in the optimal theoretical case, of an asymptotic nature.

An approach to algorithmic complexity introduced before~\cite{BDMIntro,GenerativeModels} is based upon the so called \emph{Coding Theorem Method} (CTM)~\cite{CTMShortStrings,Soler-Toscano2017-SOLACM,soler2014calculating,BDMIntro}, which relies on a relation between the algorithmic prefix complexity of a string and its algorithmic probability~\cite{IKCA,Levin1974}. The method approximates the value of a string's algorithmic complexity via the frequency distribution of that string as an output of small Turing machines and other models. Zenil and Delahaye~\cite{Zenil_OnTheAlg} pioneered this empirical approach to algorithmic complexity based on output frequency distributions across different computational models (e.g. Turing machines, cellular automata, and Post systems) and with empirical distributions from real-world/physical data. 

A criticism of CTM~\cite{Vitanyi_How_Incomp} is its potential dependency on the reference machine, which may affect the stability of the method in the face of changes in the  reference model. The other only alternative based on popular lossless compression does not address this problem either and is more related to classical Information Theory (Shannon Entropy) than to Kolmogorov complexity~\cite{ZenilRespuesta}. The findings of Zenil et~al.~\cite{Zenil_OnTheAlg} revealed a tendency towards simplicity occurring with greater frequency and with some alignment across models. As CTM depends on the outputs from the execution of a subset of programs for a particular model of computation, changing the subset of programs or the reference model could lead to inconsistent results, even in compliance with the theory. However, the theory guarantees convergence and in this work we explore how slow or fast a higher-level programming language may converge.

The authors of CTM have reported empirical evidence in favour of the stability of complexity approximations in the face of changes in the computational formalism upon application of CTM~\cite{CTLB2019,ZenilRespuesta,Zenil_OnTheAlg,Zenil_CompressionValidation}, further suggesting a convergence towards some early natural distribution based on the complexity of the underlying model of computation~\cite{DelahayeTowardsStableDefinition} even without the guarantee of optimality. The conjecture in~\cite{DelahayeTowardsStableDefinition} can be summarised as implying that the simpler (shorter) a model of computation the faster it may converge. The early reported correspondence between high algorithmic probability strings from output frequency distributions has been suggested as a defining characteristic of `natural' computational models~\cite{DelahayeTowardsStableDefinition}, specifically those aligning with the distributions observed in Turing machines and cellular automata. This led to the conjecture that simpler models converge faster than more complex or artificial ones.

In this paper we introduce a computational model that aspires to meet both desirable theoretical constraints and practical applicability. To begin with, the model is optimal. Additionally it is based on a high level programming language which, hopefully, can offer shorter versions of program routines performing abstract tasks. The optimal universal interpreter for our language, called IMP2, is constructed in such a way that a well-defined Kolmogorov prefix complexity is specified. This entails both a suitable choice for a binary prefix-free encoding of programs and an efficient enumeration of programs. Moreover, given that our proposed model is conformant with the theory of algorithmic prefix complexity, it can be used for an independent empirical test of the previously suggested hypothesis regarding the convergence towards a `natural' distribution upon application of CTM.

The computational methods implemented allow us to generate all IMP2 programs up to a given program length, as well as executing each program in a resource-bounded evaluator in order to obtain the output strings. Our framework allows to produce approximations to Kolmogorov complexity of a relatively small set of strings, using both CTM and a more direct approximation based on the length of the smallest program found in the program space, which we call SPF. This was done with very modest computational means just to get a taste of what can be done with the language before embarking on larger scale calculations. 

Our results show that the degree of correlation between the estimated CTM complexity under IMP2 and the previous estimation values published in the literature~\cite{CTMShortStrings,soler2014calculating} vary widely and significantly depending on the scale of the data being considered. The differences we found in complexity estimations under IMP2 and other models of computation indicate that some models might not conform to the algorithmic `natural' behaviour observed previously~\cite{CTLB2019,ZenilRespuesta,Zenil_OnTheAlg}. This further suggests that the CTM methodology is in fact sensible to changes in the reference model chosen even if such a model is not biased towards being not optimal or reasonable. However, surprisingly, we also found (a)~a correlation between the approximations via CTM and the length of the smallest program found (SPF) for each string under our own reference machine and (b)~that SPF is a coarse-grained measurement of~CTM. These two facts further support the claims of the validity of CTM as an approach to algorithmic complexity presented in previous reports~\cite{Zenil_Correspondence}, even if the empirical approximations produced are not exempted from the theoretical drawbacks arising from dependence on the chosen model.

\section{Preliminaries}

{\it Algorithmic complexity}~\cite{KolmogorovOriginal}, also known as \emph{Kolmogorov complexity}, of a string of bits $x$ with respect to a universal Turing machine $U$ is defined as:
\[ K_U(x) = \min \left\{ |p| {}:{} U(p) = x \right\}, \] 
where $|p|$ denotes the length in bits of the program $p$, and $U(p) = x$ means that the universal Turing machine $U$, given input $p$, produces output $x$ upon completing its execution. On the other hand, \emph{prefix complexity} is a variant of plain Kolmogorov complexity as defined above where we require the reference Turing machine $U$ to be \emph{prefix-free}, that is, the set of programs it can execute forms a prefix-free set (meaning no program is a prefix of another).

\emph{A machine $U$ is additively optimal}~\cite{IKCA} for a class of machines (say the class of Turing machines or prefix machines), if for any other machine $U'$ of the same class, there is a constant $c$ (depending on $U$ and $U'$ but not on $x$) such that for all $x$:
\begin{equation}\label{eq:kolmogorov-def}
    K_{U}(x)  \leq K_{U'}(x) + c
\end{equation}
Moreover, the {Invariance theorem} guarantees the existence of additively optimal machines for both Turing machines and prefix machine classes. In what follows, for simplicity, we will refer to additively optimal machines as being optimal.

The \emph{algorithmic probability} (also known as the universal distribution)~\cite{ZvonkinLevin1970} of a string $x$ with respect to a Turing machine $U$ is defined as:
\[ m_U(x) = \sum_{U(p)=x} 2^{-|p|} \]

CTM~\cite{CTMShortStrings} intends to produce an empirical approximation of algorithmic probability by obtaining the output strings from the execution of all Turing machines with $n$ states and $m$ symbols, i.e. the set $(n,m)$. Then, the relative frequency, $D(n,m)(x)$, of the output string $x$, with respect to the total number of halting machines in $(n,m)$, is taken as an approximation to its algorithmic probability:
\begin{equation}\label{eq:freq}
D(n,m)(x) = \frac{\vert \left\{ p\in[1, \vert(n,m)\vert] : T_{p}(\epsilon) = x \right\} \vert}{\vert \left\{ p\in[1, \vert(n,m)\vert] : T_{p} \textrm{ halts} \right\} \vert},
\end{equation}
where $T_{p}$ denotes the Turing machine with number $p$ in the set $(n,m)$, according to the enumeration in~\cite{CTMShortStrings}. 

The relationship between algorithmic probability and prefix Kolmogorov complexity is established by the \emph{Coding Theorem}~\cite{IKCA}, which states that there is a constant $c$, such that for every string $x$,
\[ |-\log_2 m_U(x) - K_U(x)| < c. \]

Based on this relation, given that the constant bound is asymptotically non-significant, CTM approximates the complexity of a string $x$ by 
\begin{equation}\label{eq:ctm}
CTM_{(n,m)}(x)=-\log_2 D_{(n,m)}(x).
\end{equation}

\section{Methodology and techniques}

One of the main goals in our study was to obtain a reference computational model meeting some desirable constraints, in particular, being optimal and prefix free. We call this reference machine IMP2, as it is the result of mayor modifications to IMP, a very well known small imperative language used in~\cite{LEMUS202231}, with the purpose of achieving provable prefix-free optimality.

In order to adapt the CTM to our IMP2, we enumerate its programs by increasing length and execute them on a resource-bounded interpreter. When the number of execution steps exceeds a given threshold~\cite{LEMUS202231} the program is considered non-halting. When a program halts, an output string is computed and the relative frequency of the output is adjusted accordingly. Frequency values for each output string are then used to obtain the CTM approximation to algorithmic complexity under the IMP2 model. 

Additionally, the length of the first program producing a string is also registered as its SPF (\emph{Smallest Program Found}) complexity approximation. Of course, the insolvability of the halting problem gets in the way of considering SPF's length as the Kolmogorov complexity of the string under IMP2.

\subsection{From IMP to IMP2}

In~\cite{LEMUS202231}, the authors present a proof of concept of an approach to approximate plain Kolmogorov complexity using SPF, based on a reference machine called IMP. Although this is a reasonable computational model based on an imperative high-level programming language, it's not straightforward to prove weather or not this model is optimal. The same is also true for other similar attempts to approximate Kolmogorov prefix complexity~\cite{CTMShortStrings,soler2014calculating}, in which the approach is based on a reference universal Turing machine with a reasonable or `natural' design whose optimality is also merely conjectured.

As for the prefix-free property, since the Coding Theorem is central to CTM, which is one of the approximation methods that concerns us in this study, we require our chosen reference machine to be prefix-free. However, we can hardly see IMP, as it is used in the approach presented in~\cite{LEMUS202231}, as a suitable reference model for a well-specified prefix Kolmogorov complexity. This is mainly because the approach to `program length' used in~\cite{LEMUS202231} is associated with the number of nodes in the abstract syntax tree of the corresponding sentence of the high-level (programming) language grammar. Even though a prefix free encoding could be devised for the IMP language programs, there is no well-defined binary encoding of `programs' in the approach presented in~\cite{LEMUS202231}, impeding  a direct interpretation of program length as the number of bits in its binary representation.  The presence of such binary encoding for the input programs of the reference machine is crucial for this length interpretation, as well as for satisfying the prefix-free property, a fundamental requirement for a well-specified Kolmogorov prefix complexity. The specification of such encoding for the input programs is precisely one of the most significant differences between our current model proposal, IMP2, and the one presented in~\cite{LEMUS202231}.

The development of the IMP2 model was predominantly driven by the aim of creating a reference machine suitable for approximating Kolmogorov prefix complexity while adhering to the fundamental design principles of the original IMP framework. This led us to the creation of a fundamentally distinct model suitable for proving optimality and that preserves the syntactic structure, output conventions, and semantic properties of the IMP language.

\subsection{Syntax of IMP2}
\label{sec:imp2lang}

IMP2 is a small high level imperative programming language. The language's syntax is specified by the context-free grammar in Figure~\ref{fig:grammar}.

A {\it well specified} input program consists of a two-part code $\langle n, y \rangle$ defined as
\begin{equation}\label{eq:twopartcode}
    \langle n, y \rangle = \bar{n}\cdot y
\end{equation}
where $n$ specifies a particular sentence $P_n$ of IMP2 by its index $n$ in an enumeration of all valid sentences of the language (see Section~\ref{sec:enum}) and $\bar{n}$ denotes its self-delimiting binary encoding. For the latter, consider the string $\mathcal{B}_n$ according to the length-increasing lexicographic enumeration of all binary strings:
\[ (\epsilon, 0), (\mathtt{0}, 1), (\mathtt{1}, 2), (\mathtt{00}, 3), 
(\mathtt{01}, 4), (\mathtt{10}, 5), (\mathtt{11}, 6), \dots, (\mathcal{B}_n, n), \dots\]
Then $\bar{n} = 1^{|\mathcal{B}_n|} \cdot 0 \cdot \mathcal{B}_n$ (the prefix $1^{|\mathcal{B}_n|} \cdot 0$ is for making it self-delimiting). $y$ is a binary input stream.
\begin{figure}
    \begin{align*}
      P &\rightarrow \texttt{skip} \mid \texttt{x[} N \texttt{] := } A  \mid \texttt{(while } B \texttt{ do } P\texttt{)} \mid \texttt{(} P \texttt{ ; } P\texttt{)}  \mid \texttt{(if } B \texttt{ then } P \texttt{ else } P\texttt{)} \\
      B &\rightarrow \texttt{true} \mid \texttt{false} \mid \texttt{(}A \texttt{ = } A\texttt{)} \mid \texttt{(}A \texttt{ < } A\texttt{)} \mid \texttt{(}B \texttt{ and } B\texttt{)} \mid \texttt{(}B \texttt{ or } B\texttt{)} \mid \texttt{not } B \\
      A &\rightarrow \texttt{readbit} \mid N \mid \texttt{x[} N \texttt{]} \mid \texttt{(}A \texttt{ + } A\texttt{)} \mid \texttt{(}A \texttt{ - } A\texttt{)} \mid \texttt{(}A \texttt{ * } A\texttt{)}
    \end{align*}
    \caption{Context-free grammar for IMP2. $N$ stands for natural numbers without leading zeros.}\label{fig:grammar}
\end{figure}

\subsection{Semantics of IMP2}
The semantics is based on an array of memory locations and a potentially empty stream of input bits. The execution of an IMP2 program $\langle n, y \rangle$ starts by scanning and parsing the $n$-part of the input program. A second step is identifying the $n$-th IMP2 sentence $P_n$, according to the enumeration presented in Section~\ref{sec:enum}. Then, the machine sets all memory locations to $0$ and prepares to read the input stream $y$. Then the interpreter executes $P_n$ with $y$ in the input tape. Commands have the usual meaning in standard high-level programming languages.

The special arithmetic operation \texttt{readbit} allows programs to read a single bit from the input stream, starting from the first bit of the string $y$, this being the only way of scanning the input. The next \texttt{readbit} operation will read the following bit, meaning bits are read in order and cannot be read twice. If the program reads from an empty input stream it loops indefinitely and thus is considered a non-halting program. An assignment of \texttt{readbit} to a memory location stores the bit read in that location. 

Memory locations are indexed by non-negative integers and contain arbitrarily large non-negative integers. Figures~\ref{fig:imp2exA} and \ref{fig:imp2exB} are examples of IMP2 sentences.

A statistical algorithm, based on sampling the program space to estimate a halting threshold~\cite{CALUDE2020}, was incorporated to the IMP2 interpreter in order to ensure termination and mitigate the effects of the halting problem~\cite{LEMUS202231}.

IMP2 has all necessary features to compute any partial computable function and it is therefore Turing complete.

\begin{figure}
    \centering
\begin{verbatim}
(x[0] := 5;
 (x[1] := 1;
  (while (0 < x[0]) do
    (x[1] := (x[1] * x[0]);
     x[0] := (x[0] - 1)))))
\end{verbatim}
    \caption{Example of an IMP2 sentence computing $5!$ and storing the result in location $1$.}
    \label{fig:imp2exA}
\end{figure}

\begin{figure}
    \centering
\begin{verbatim}
(x[0] := readbit;
 (while (readbit = 1) do
   x[0] := (x[0] + 1)))
\end{verbatim}
    \caption{Example of IMP2 sentence counting consecutive $1$ bits from the input stream and storing the result in location $0$.}
    \label{fig:imp2exB}
\end{figure}

\subsection{Input/output convention}

A {\it valid} halting IMP2 program $\langle n, y \rangle$ is such that the machine halts after reading exactly all the bits from the input stream. For example, if the execution of the program $\langle n, 001 \rangle$ halts after invoking the \texttt{readbit} operation only twice, this means that $\langle n, 001 \rangle$ is not a valid IMP2 program, but $\langle n, 00 \rangle$ is.

The result of the execution of a halting program is determined by the state of the memory locations. To be useful for the purpose of approximating algorithmic complexity, the contents of the memory locations must be interpreted as a particular binary string. Although many interpretations could be put in place, the chosen one is the concatenation of the binary strings corresponding to the values in each of the memory locations of a program. The conversion of non-negative integers to binary strings is according to the previous convention and the concatenation of strings is performed in increasing order of the subscripts of the memory locations.
 
As an example, consider the sample program for computing $5!$ in figure~\ref{fig:imp2exA}. After the program halts the values stored in memory locations are $\texttt{x[0]} = 0$ and $\texttt{x[1]} = 120$, with other locations containing their initial values of $0$, which are mapped to the empty string $\epsilon$. Thus, since $\mathcal{B}_\texttt{x[0]} = \epsilon$ and $\mathcal{B}_\texttt{x[1]} = \texttt{111001}$, the resulting output string is $\epsilon\cdot\texttt{111001} = \texttt{111001}$.

Thus, we say that IMP2 produces the string $x$ as \emph{output}, with \emph{input} $\langle n, y \rangle$, and write $\text{IMP2}(\langle n, y \rangle) = P_n (y) = x$, if the segment of the input scanned at the moment of halting is exactly the string $\langle n, y \rangle$, and $x=\mathcal{B}_\texttt{x[0]} \cdot \mathcal{B}_\texttt{x[1]}\dots$.

\subsection{Prefix-free optimality}

Prefix-free machines, or simply prefix machines, are also known in the literature as \emph{self-delimiting} machines~\cite{Chaitin2001,IKCA}. Under our convention for inputs and outputs, it is not difficult to see that IMP2 is a universal prefix machine. Notice that the two-part code $\langle n, y \rangle$ of an IMP2 program uses the previously defined self-delimiting encoding $\bar{n}$. Moreover, IMP2 is forced to accept an input program on its own, without using an end-of-string indication for when to stop scanning the input tape. This ensures that the set of strings accepted, i.e. the set of halting programs, is prefix-free.

\begin{proposition}\label{prop:optimal}
The universal machine $\text{IMP2}$ is additively optimal for the class of prefix machines.
\end{proposition}
\begin{proof}
Consider an arbitrary prefix machine $G$ over a binary input alphabet. Then, since IMP2 is a universal machine, there is an IMP2 sentence $P_n$, such that, for any binary string $y$:
\[\text{IMP2}(\langle n,y \rangle) = P_n (y) = G(y)\]
Hence, considering that $\bar{n} = \mathtt{1}^{|\mathcal{B}_n|} \cdot\mathtt{0}\cdot \mathcal{B}_n$ and $|\bar{n}|=2(|\mathcal{B}_n|)+1$, it follows from the definition of the algorithmic complexities $K_{\text{IMP2}}$ and $K_G$ that:
\[K_{\text{IMP2}}(x)  \leq K_G(x) + 2(|\mathcal{B}_n|) + 1,\] 
where $2(|\mathcal{B}_n|) + 1$ is a constant depending only on IMP2 and $G$. Therefore, it follows that IMP2 is an additively optimal machine for the class of prefix machines.
\end{proof}

\subsection{Enumeration of IMP2 sentences}
\label{sec:enum}

Sentences of the IMP2 language are enumerated based on the structure of syntactically valid abstract syntax trees (AST).  Using standard methods for enumerating combinatorial structures~\cite{1994FlajoletCombinatorialStructures,TAOCP4a} we implemented a system for constructing enumerators based on combinators~\cite{LEMUS202231,2017NewFetscherFindlerMccarthyFairEnumerations} that allow us to build and compose complex enumerations from simpler ones.  The enumerations of numbers, truth values and other atomic elements are combined to produce enumerations of arithmetic and boolean expressions, which are themselves combined to produce enumerations of assignments and control structures all the way up to the top level sentences specified by the language grammar.

The bijection between sentences and the natural numbers is encoded with a pair of functions. The rank function produces the unique position of a given AST, and the unrank function is the inverse, producing a unique AST given its position.

The most important method of combination consists of getting the product of two or more enumerations, where the positions are interweaved using Cantor's pairing function in the binary case and a generalisation for $n$-tuples of natural numbers which satisfies the
fairness property~\cite{2017NewFetscherFindlerMccarthyFairEnumerations}.

When combining enumerations for a particular syntactic category such as arithmetic expressions, boolean expressions or IMP2 sentences, the positions are partitioned in an alternating fashion between the production rules.

Considering the program for calculating $5!$ of Figure~\ref{fig:imp2exA} and the enumeration of IMP2 sentences, the assignment $\texttt{x[0] := 5}$ is mapped to position $1405$, the more complex assignment $\texttt{x[1] := (x[1] * x[0])}$ is mapped to position $142049$, the while loop in the example is mapped to position $17972673899864641600766$, and the whole example sentence is mapped to a position whose decimal representation contains 90 digits.

\subsection{Enumerating IMP2 programs by length}

The enumeration of programs consisting of a sentence and its input stream must take into account the number of bits in the corresponding two-part code $\langle n,y \rangle$ in order to produce a length-increasing order.

For a given code length $m$, we consider all possible combinations of $k$-bit self-delimiting strings $\bar{n}$ and $(m-k)$-bit binary strings $y$ for $0 \leq k \leq m$.  This can be easily accomplished given the simple structure of the encoding, yielding $2^{m-k}$ binary strings for the input stream and $2^{(k-1)/2}$ sentences when $k$ is an odd number and zero programs otherwise.

The subset of programs of a particular length can be partitioned and enumerated independently, allowing for a distributed enumeration and execution of programs.

\section{Results}\label{sec:results}

Our preliminary experiment was performed by executing all IMP2 programs of at most length $40$, denoted $\text{IMP2}_{40}$. Running this experiment took $34$ minutes total for estimating the halting threshold and then $51.7$~hours for the enumeration and execution of $\text{IMP2}_{40}$. We used a desktop computer dedicating $30$~threads of an AMD Ryzen 9 CPU running at $4.5$~GHz with $6$~GiB of memory each.

We recorded the length of the smallest program found for every output string $x$, denoted as $\text{SPF}_{40}(x)$. Additionally, for approximating the complexity of strings via CTM, we recorded the relative frequency of programs producing every output string $x$ with respect to the total number of halting programs, denoted $D(\text{IMP2}_{40})(x)$.  The full dataset containing the complexity estimations can be found in~\url{https://kolm-complexity.gitlab.io/optimal-interpreter/}.

For putting in perspective the scale of our experiment we use as references the empirical frequency distributions $D(4,2)$ and $D(5,2)$, as published in~\cite{CTMShortStrings,soler2014calculating}. Table~\ref{tab:experiment-summary} shows a summary contrasting some key aspects of the experiments. In the table, the entry corresponding to `complete output length' records the length up to which all strings of that length and below were produced and accounted for in the distributions.
\begin{table}
    \caption{Summary of enumerating and executing programs.}\label{tab:experiment-summary}
    \centering
    \begin{tabular}{|r|l|l|l|}
    \hline
                        & $\text{IMP2}_{40}$ & $(4, 2)$ & $(5, 2)$ \\
    \hline
       Total programs   & $2199020109825$ & $11019960576$ & $26559922791424$ \\
      Strings produced  & $145$      & $1832$      & $99608$ \\
 Largest output length  & $7$        & $16$        & $49$    \\
Complete output length  & $6$        & $8$         & $12$    \\
      \hline
    \end{tabular}
\end{table}

Table~\ref{tab:halt-status} shows the counts and percentages of programs in $\text{IMP2}_{40}$ according to their halting status. The vast majority of programs are considered non-halting. However, most of them did run up to termination but failed to read all the bits in the input stream, which means they are extensions of a previously executed program corresponding to the same IMP2 sentence but where a smaller input stream, i.e. a prefix, was accounted for already.
\begin{table}
    \caption{Halting status for $\text{IMP2}_{40}$}\label{tab:halt-status}
    \centering
    \begin{tabular}{|l|l|l|}
    \hline
    Status & Count & Percentage \\
    \hline
    Halted & $3828109$ & $0.0000017\%$ \\
    Threshold surpassed & $887841$ & $0.000040\%$ \\
    Loop reached & $392291668910$ & $17.839385\%$ \\
    Read failure & $45125277$ & $0.002052\%$ \\
    Extensions & $1806678599688$ & $82.158348\%$ \\
    \hline
    \end{tabular}
\end{table}

Although the scale of our experiment is very modest, the analysis of the data that we were able to collect allows us to make specific observations in relation to previously suggested properties of CTM complexity approximations.

\subsection{On the convergence towards a `natural' distribution}

In other CTM-related studies, empirical evidence has been reported regarding the correspondence between the CTM complexity rankings produced by different models of computation~\cite{CTLB2019,ZenilRespuesta,Zenil_OnTheAlg,Zenil_CompressionValidation}, alluding to an underlying natural distribution of binary strings. In particular, the Spearman correlation coefficient has been employed for assessing the correlation between the frequency distributions of binary strings generated by Turing machines, one-dimensional cellular automata and Post tag systems~\cite{Zenil_OnTheAlg}, suggesting a sort of `natural behaviour', defined in~\cite{CTLB2019} as ``behaviour that is not artificially introduced with the purpose of producing a different-looking initial distribution''. Moreover it has been suggested the hypothesis that most `natural' models produce similar output frequency distributions~\cite{ZenilRespuesta}, which means some kind of empirical invariability in the complexity estimations by CTM even without a guarantee of optimality.

To the effect of testing such hypothesis, we believe the IMP2 model to be a suitable choice for a reference machine, given its neutrality, in the sense that, while representing a significant departure from the previously chosen `natural' models, is still reasonable and even optimal. In other words, although the major differences between IMP2 and these other models giving raise to such conjecture suggest that IMP2 is unlikely to be biased towards producing a similar output frequency distribution, IMP2 is also not designed in an artificial way with the purpose of producing a different-looking initial distribution. Therefore, we believe this neutral model to present a valuable opportunity to test the hypothesis regarding the convergence towards a `natural' distribution upon application of CTM.

In our study, the Spearman rank correlation coefficient was also employed to contrast the frequency distribution $D(\text{IMP2}_{40})$, with the distributions $D(4,2)$ and $D(5,2)$ as published in~\cite{CTMShortStrings}. For assessing the statistical significance of the correlation values we performed a permutation test and considered $p$-values in the ranges $[0,0.001)$, $[0.001,0.01)$ and $[0.01,0.1)$ to represent very high, high and low significance, respectively, and $0.1$ or larger to represent very low significance.

The results from this study reveal a striking difference when analyzing the data at different resolutions. This is, globally, when considering all CTM values computed, and locally, by zooming in on the values for strings of fixed lengths. Upon examining the whole dataset, we found a global correlation of $0.850543$ with very high significance between $\text{IMP2}_{40}$ and both $(4,2)$ and $(5,2)$~\cite{CTMShortStrings, soler2014calculating}. Figure~\ref{fig:ctm-imp2-d4} shows the distribution of the approximated CTM values of each string between $\text{IMP2}_{40}$ and $(4,2)$. However, when looking at the local Spearman coefficients between $\text{IMP2}_{40}$ and both $(4,2)$ and $(5,2)$, we found no correlation, regardless of the length of the strings considered. Altogether, these observations indicate that the reason behind the global high correlation observed is a sheered commonality between the three distributions that has no impact in the local correlation analysis for the strings of a fixed length. Such commonality is that in all three distributions most strings of a given length are less frequent (more complex) that all strings of a smaller length. 

\begin{figure}
\centering
\includegraphics[width=0.7\textwidth]{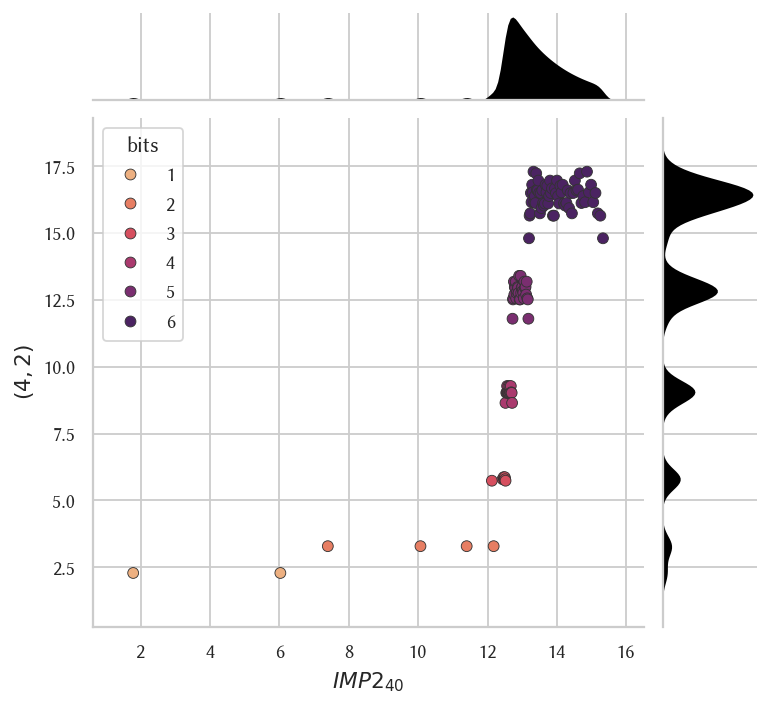}
\caption{%
  CTM approximation comparison between $\text{IMP2}_{40}$ and $(4,2)$.} \label{fig:ctm-imp2-d4}
\end{figure}

While both global and local correlation perspectives provide insights into the compatibility of CTM approximations between reference machines, the discrepancies found suggest that the IMP2 model, however Turing universal and optimal, might not conform to the natural behaviour defined recently in~\cite{DelahayeTowardsStableDefinition} given its peculiarities, including the fact that the model does not display conformance to basic linear transformations (binary negation, for example), since the output distribution produced does not converge to previously reported `natural' distributions and may converge slower to Levin's universal distribution.

\subsection{Validation of CTM by SPF}

Since our methods for computing the approximations enumerate and execute programs by length in increasing order, it is equally straightforward to count the number of programs producing a string for CTM, as it is to record the length of the smallest program found that produces a string for SPF. Furthermore, the prefix-free encoding of our IMP2 model yields programs at almost every length, in contrast to previously used models such as Turing machines with $n$ states for which programs are all encoded as strings with $\lceil\log_2{\left((4n+2)^{2n}\right)}\rceil$ bits~\cite{MAAC}, allowing us to run programs of a larger range of lengths. Compared to CTM, the SPF approximation resembles more the original definition of Kolmogorov complexity of equation~\ref{eq:kolmogorov-def} and therefore can be considered a suitable approximation as it is a proper upper bound for the actual prefix complexity, and thus an alternative measure suitable for validating the corresponding CTM approximation.

When contrasting the algorithmic complexity approximated using CTM and SPF under IMP2 (see Figure~\ref{fig:plot2}) we get a Spearman correlation coefficient of $0.986114$ and a Pearson correlation coefficient of $0.911119$ both of very high significance. We also found high and significant correlation coefficients between both measures when looking at strings of a particular length up to 6 bits (see Figure~\ref{fig:ctm-spf-6}). In other words, we found that both measures are highly correlated from a global and local perspectives. Table~\ref{tab:correlations} details all Spearman correlation coefficients and $p$-values between the frequency distributions produced by the different aforementioned reference machines and also between CTM and SPF using IMP2.

\begin{figure}
    \centering
    \includegraphics[width=0.7\textwidth]{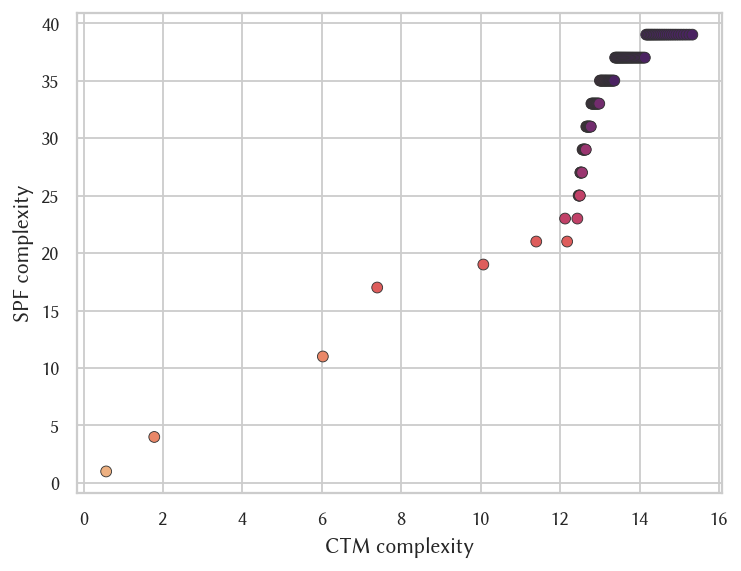}
    \caption{Complexity estimations comparison for $\text{IMP2}_{40}$ for all binary strings of length $6$ and below.}
    \label{fig:plot2}
\end{figure}

\begin{figure}
    \centering
    \includegraphics[width=0.7\textwidth]{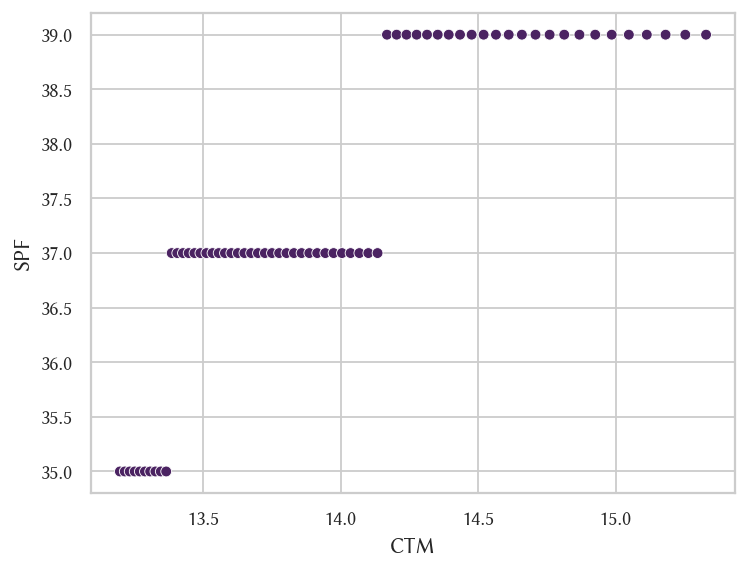}
    \caption{Complexity approximation comparison for $\text{IMP2}_{40}$ for all $6$-bit strings.}
    \label{fig:ctm-spf-6}
\end{figure}

\begin{table}
    \caption{Spearman correlation coefficients of the frequency of output strings between reference machines (first two columns) and complexity approximations with IMP2. The numbers in parenthesis indicate the correlation significance computed with a permutation test of $20000$ samples given the null hypothesis of a random ranking yielding a higher correlation.}\label{tab:correlations}
    \centering
    \begin{tabular}{|l|l|l|l|}
    \hline
    Output length & 
    IMP2 vs $(4,2)$ & 
    IMP2 vs $(5,2)$ & 
    CTM vs SPF \\
    \hline
    3 & 
    $0.0$ ($0.55622$) & 
    $0.0$ ($0.56147$) & 
    $0.92582$ ($0.00264$) \\
    4 & 
    $0.0$ ($0.50062$) & 
    $0.0$ ($0.49927$) & 
    $0.91715$ ($4.99975\mathrm{e}{-5}$) \\
    5 & 
    $0.0$ ($0.49837$) & 
    $0.0$ ($0.50162$) & 
    $0.91721$ ($4.99975\mathrm{e}{-5}$) \\
    6 & 
    $0.0$ ($0.49832$) & 
    $0.0$ ($0.49667$) & 
    $0.91687$ ($4.99975\mathrm{e}{-5}$) \\
    \hline
    $\leq 6$ & 
    $0.85052$ ($4.99975\mathrm{e}{-5}$) & 
    $0.85052$ ($4.99975\mathrm{e}{-5}$) & 
    $0.98577$ ($4.99975\mathrm{e}{-5}$) \\
    \hline
    \end{tabular}
\end{table}

Although we can not confidently say this correlation is not dependent in some way from the chosen reference model, we have no reason to believe that our model is in any way artificially biased to that effect. Therefore, such correlation result comes as a positive surprise since it is beyond the scope of theoretical expectations, because of the way that the asymptotic term from the Coding Theorem is handled by CTM. This supports the claim that CTM is suitable for comparing the relative complexities between output strings. Additionally, our results show that, in most cases, CTM allows us to set apart strings we are not able to distinguish by its SPF value, thus providing a finer-grade approximation. Moreover, a related previous report~\cite{Zenil_Correspondence}, upon execution of a large set of Turing machines with up to 5 states (and 2 symbols), showed that the CTM approximation of the complexity of a string is in agreement with the number of instructions used by the machines producing the string. These considerations give us confidence in using CTM as a valid substitute for the more naive approach.

\section{Concluding remarks}

In this paper we presented a universal optimal model of computation with a simple high-level imperative programming language called IMP2. This model served as a reference machine for approximating Kolmogorov prefix complexity by means of the Coding Theorem Method (CTM). In contrast to previous studies that employed low-level models as reference machines whose optimality is merely conjectured, our IMP2 model was developed with theoretical constraints in mind. This approach enabled us to test previously published hypotheses regarding empirical approximations of complexity using a reasonably constructed reference machine.

A first hypothesis tested is if CTM approximations would yield similar numerical results when computed with different reference machines, indicating that the approximation might be invariant with respect to the model of computation chosen, when `natural'. The observed output frequencies computed using IMP2 are strongly aligned with the length-increasing lexicographic ordering of binary strings, and while the frequencies obtained with other models have a somewhat similar tendency to simplicity bias, for fixed length strings, shows little to no correlation. These results suggest that there may be models that are more `natural' than others even when equally optimal or not~\cite{DelahayeTowardsStableDefinition}. This means that IMP2 would require a substantially larger space of programs to be executed in order to observe such a convergence. It also remains open if there may be choices we made for the design of IMP2 that inadvertently makes it less `natural' and unsuitable for approximating Kolmogorov complexity under CTM (and other means). There needs to be further investigations into the sensitivity of CTM to changes in the choice of prefix-free encoding, enumeration, and the input/output convention, as well as with other reasonable high-level models of computation.

The second hypothesis poses that CTM is a valid approach for approximating Kolmogorov complexity. To verify this, we compared CTM with SPF, a direct application of the Kolmogorov prefix complexity definition. Our results demonstrate that both approximations are monotonically correlated, meaning there is an agreement in the order of assigned complexity values. This indicates that our empirical CTM complexity approximation is stable and meaningful with respect to a fixed reference machine. In addition, there is a practical advantage in choosing CTM, as it allows for comparing complexities with a finer-grade resolution than SPF. Both approximation methods can be further studied to estimate a lower bound for the constant in the coding theorem that is discarded by CTM.

One of the main challenges we faced in our endeavour is the overwhelming proportion of non-halting programs encountered while executing the experiments with IMP2, resulting in a significant amount of executions performed that do not contribute towards the calculation of either CTM or SPF. As shown in Table~\ref{tab:halt-status}, most executions are considered non-halting, and among these, extension programs constitute the majority. It is not straightforward to preemptively omit these subset of programs from the enumeration or efficiently detect them before running them. This disadvantage is not inherent to the approximation method, as it doesn't present itself with other reference machines, but it is mainly due to the prefix-free code nature of IMP2 programs, which is an essential component for our optimality proof. We posit there is a trade-off between practicality and optimality when empirically approximating complexity with the present approach.

\printbibliography

\end{document}